\documentclass[10pt,a4paper]{amsart}

\usepackage{hyperref}
\usepackage{latexsym,amsmath, bm}
\usepackage{enumerate}
\usepackage{amsfonts}
\usepackage{amssymb}
\usepackage{geometry}
\usepackage{latexsym}
\usepackage{fixmath}
\usepackage{faktor}
\usepackage{mathtools}
\usepackage[T1]{fontenc}
\usepackage{fourier}
\usepackage{bbm}
\usepackage{color}
\usepackage{pgfplots}
\usepackage{tikz}
\usetikzlibrary{shapes,decorations,arrows,calc,arrows.meta,fit,positioning}
\usepackage{graphicx}
\usepackage{fullpage}
\usetikzlibrary{decorations.pathreplacing, matrix}
\usepackage[boxed, linesnumbered, noend, noline]{algorithm2e}

\numberwithin{equation}{section}

\newcommand{\vA}{\vec A}
\newcommand{\vB}{\vec B}

\newcommand{\vX}{\vec X}
\newcommand{\vY}{\vec Y}

\newcommand\cE{\mathcal{E}}

\newcommand\cO{\mathcal{O}}

\newcommand{\Po}{{\rm Po}}
\newcommand{\Bin}{{\rm Bin}}

\renewcommand{\vec}[1]{\boldsymbol{#1}}
\newcommand{\vecone}{\vec{1}}

\newcommand\NN{\mathbb{N}}
\newcommand\pr{\mathbb{P}} 
\renewcommand\Pr{\pr}

\newcommand\Erw{\mathbb{E}}

\newcommand{\whp}{w.h.p.}

\newtheorem{definition}{Definition}[section]

\newtheorem{theorem}[definition]{Theorem}
\newtheorem{lemma}[definition]{Lemma}

\newtheorem{fact}[definition]{Fact}

\newcommand\Lem{Lemma}

\newcommand\Thm{Theorem}

\newcommand\bc[1]{\left({#1}\right)}
\newcommand\cbc[1]{\left\{{#1}\right\}}

\newcommand\abs[1]{\left|{#1}\right|}

\newcommand{\zero}{V_0}
\newcommand{\one}{V_1}

\newcommand{\zeroplus}{V_0^+}
\newcommand{\zerominus}{V_0^-}
\newcommand{\oneplus}{V_1^+}
\newcommand{\oneminus}{V_1^-}

\newcommand{\onei}[1]{V_{1,({#1})}}
\newcommand{\zeroplusi}[1]{V_{0,({#1})}^+}
\newcommand{\zerominusi}[1]{V_{0,({#1})}^-}
\newcommand{\oneplusi}[1]{V_{1,({#1})}^+}
\newcommand{\oneminusi}[1]{V_{1,({#1})}^-}

\newcommand{\G}{\mathbb{G}}

\begin{document}
	
	\title{Note on the offspring distribution for group testing in the linear regime}
	
	\thanks{Oliver Gebhard and Philipp Loick are supported by DFG 646/3}

	\author{Oliver Gebhard, Philipp Loick}

    \address{Oliver Gebhard, {\tt gebhard@math.uni-frankfurt.de},  Goethe University, Mathematics Institute, 10 Robert Mayer St, Frankfurt 60325, Germany.}
    \address{Philipp Loick {\tt loick@math.uni-frankfurt.de}, Goethe University, Mathematics Institute, 10 Robert Mayer St, Frankfurt 60325, Germany.}

	\begin{abstract}
	The group testing problem is concerned with identifying a small set of $k$ infected individuals in a large population of $n$ people. At our disposal is a testing scheme that can test groups of individuals. A test comes back positive if and only if at least one individual is infected.
	In this note, we lay groundwork for analysing belief propagation for group testing when $k$ scales linearly in $n$. To this end, we derive the offspring distribution for different types of individuals. With these distributions at hand, one can employ the population dynamics algorithm to simulate the posterior marginal distribution resulting from belief propagation.
	\end{abstract}
	
	\maketitle

\section{Introduction}\label{sec_intro}

\subsection{Motivation and related work}

The group testing problem, originally introduced by Dorfman \cite{Dorfman_1943}, is a prime example of a statistical inference problem. The problem can be formulated as follows. Consider a population of $n$ individuals $k$ of which suffer from a rare decease. Furthermore, there is a test scheme available that can test groups of individuals.
A test returns positive if and only if there is at least one infected individual in the tested group. The goal is to recover all infected individuals within the total population with the smallest number of tests with high probability (\whp)\footnote{We say that a sequence of events $\cE_n$ occurs with high probability if $\lim_{n \to \infty} \Pr(\cE_n)=1$.}.
In the last years this problem gained significant attention (see \cite{Aldridge_2019_2} for a detailed survey) and finds application in fields such as DNA sequencing \cite{Kwang_2006,Ngo_2000}, protein interaction experiments \cite{Mourad_2013, Thierry_2006} or the current COVID-19 pandemic \cite{Goethe_2020}. 
Dorfman's original approach to recover the set of infected individuals was to perform a two-stage testing scheme. In the first stage, groups of individuals would be tested. If a test returned positive, each individual would subsequently be tested separately. Conversely, a negative test ensures that every individual in the test is uninfected and thus no further testing is needed. While this scheme improves on individual testing for low infection rates, it is clearly not optimal. Over the years, a variety of algorithms have been proposed and analysed to cut down the number of tests.
Due to its amenable properties, mathematical research on group testing focused primarily on the sublinear regime where $k \sim n^{\theta}$ for some $\theta \in (0,1)$. By today, the information-theoretic and algorithmic thresholds in this regime are well understood.
In \cite{Coja_2019_2}, Coja-Oghlan et al. provide a novel pooling scheme and a sophisticated polynomial-time algorithm achieving reliable recovery with the information-theoretically minimum  number of tests possible for one-stage and multi-stage test designs.
The natural next step of the group testing research is to move to the regime where the number of infected individuals scales linearly in the number of infected individuals (as for instance of relevance in the ongoing pandemic \cite{Covid}). Next to Dorfman's original two-stage algorithm, a number of algorithms from the sublinear regime readily carry over. First, consider a plain random regular test design where each individual is assigned to a fixed number of tests uniformly at random without replacement \footnote{A regular design is provably superior to a Bernoulli-type design in the sublinear regime and likely superior as well in the linear regime.}. A first idea would be to apply the one-stage so-called {\tt COMP} algorithm in which we classify all individuals in negative tests as uninfected and all other individuals as infected. There is a slightly more sophisticated approach called {\tt DD} under which we first classify all individuals in negative tests as uninfected as above and remove them from the test design. Then, we search for any individuals included in a positive test on its own (after removal of the definitively uninfected from before) and classify them as infected. All other individuals are classified as uninfected. Since no one-stage algorithm can recover the set of infected individuals in the linear regime \whp \cite{Aldridge_2019}, both of these approaches will not entirely solve the group testing problem, but they might serve as starting points for a first group testing stage with a follow-up stage of individual testing for those individuals that are not definitively uninfected or infected. Interestingly, the {\tt DD} algorithm closely resembles the well-known warning propagation algorithm from mathematical physics which belongs to the hand tool in solving constraint satisfaction problems. Indeed, the first two steps of the {\tt DD} algorithm are precisely what warning propagation would stipulate for the group testing problem. Some prominent applications of warning propagation in other realms include \textit{graph colouring} \cite{Molloy_2012}, \textit{constraint satisfaction problems} \cite{Achlioptas_2001}, and the \textit{k-core of a graph} \cite{Pittel_1996}. The set of definitively uninfected and definitively infected individuals in the group testing problem correspond to the \textit{hard} fields identified by warning propagation in physics jargon. The interesting question is what to do with those individuals that are not definitively uninfected or infected? Both the {\tt COMP} and {\tt DD} algorithm remain clueless. For the sublinear regime, finding the remaining set of infected individuals boils down to solving a minimum vertex cover problem and a greedy vertex cover algorithm called {\tt SCOMP} has been suggested, which, however, does not improve performance beyond {\tt DD} \cite{Coja_2019}. For the linear regime, things are trickier and the default approach of mathematicians and physicists would be to apply belief propagation to a random regular test design as described above. However, analysing belief propagation in the linear regime is a challenging endeavour due to the breakdown of central-limit-like properties.
In this note, we lay some groundwork towards analysing belief propagation in the linear regime. To be precise, we consider different types of \textit{non-hard} individuals, i.e. individuals that are not classified as definitely uninfected or infected by the warning propagation algorithm and determine their offspring distribution. We analyse the second neighbourhood of any such individual type and derive the distribution of the different types of \textit{non-hard} individuals. With these distributions at hand, one can employ the population dynamics algorithm to simulate the posterior marginal distribution resulting from belief propagation for each of the \textit{non-hard} individual types. Before we can state our result, let us introduce some notation.

\subsection{Notation}

In the following, we consider individuals $x_1,...,x_n$ out which $k$ are infected. The infection status of each individual is encoded by a configuration $\sigma \in  \cbc{0,1}^V$. We call the set of infected individuals $V_1$ and the set of uninfected individuals $V_0$. Furthermore we consider tests $a_1,...,a_m$ and denote the set of tests by $F$. We assume the number of tests conducted to scale as $m=cn$ for some constant $c>0$. With $G$ denoting a graph describing the assignment of individuals to tests and using standard graph notation such $\partial x$ and $\partial a$ for the first neighbourhood of individual $x$ or test $a$, let $\hat \sigma \in \cbc{0,1}^F$ denote the test outcomes being determined by
\begin{align*}
    \hat \sigma_a = \hat \sigma_a(G, \sigma) = \vecone \cbc{\sum_{x \in \partial a} \sigma_x > 0} \qquad (a \in F)
\end{align*}
We employ a fixed variable individual degree model. Each individual has degree $\Delta=cd/\lambda$ for some constant $d$. We restrict our attention to constants $c$ and $d$ such that $\Delta$ is integer. Each individual draws its tests from the set $F$ without replacement, leading to a fluctuating test degree $(\vec \Gamma_a)_{a \in F}$. 
This test design gives rise to different types of individuals. For uninfected individuals, we distinguish two types, $\zerominus$ and $\zeroplus$. $\zerominus$ is the set of all uninfected individuals that show up in at least one negative test. $\zeroplus$ is the counterset. Formally,
\begin{align*}
    \zerominus = \cbc{x \in \zero: \exists a \in \partial x: \hat \sigma_a = 0} \qquad \text{and} \qquad \zeroplus = \zero \setminus \zerominus.
\end{align*}
Recall from the above, that the {\tt DD} algorithm which is the application of warning propagation to the group testing problem easily classifies $\zerominus$ - the definitely uninfected. Along the same lines, let $\oneminus$ denote the set of all infected individuals that show up in at least one test where all remaining individuals are from the set $\zerominus$. Again, those individuals are easily classified as infected by warning propagation. Formally,
\begin{align*}
    \oneminus = \cbc{x \in \one: \exists a \in \partial x: \forall x \in \partial a: x \in \zerominus} \qquad \text{and} \qquad \oneplus = \one \setminus \oneminus.
\end{align*}
If we consider an individual $x$ we can divide the second neighbourhood of $x$ into four sets 
\begin{align*}
  \partial^{(2)} x=\left\{ \partial^{(2)} x\cap V_1^+\right\}\cup \left\{ \partial^{(2)} x\cap V_0^+\right\}\cup \left\{ \partial^{(2)} x\cap V_1^-\right\}\cup \left\{ \partial^{(2)} x\cap V_0^-\right\}  
\end{align*}
To shorten notation, we refer to these sets as $\oneplusi{2}, \zeroplusi{2}, \oneminusi{2}, \zerominusi{2}$, respectively. Let $\cO_{0+}$ denote the event that the considered individual $x$ is from $\zeroplus$ and $\cO_{1+}$ the event that it is from $\oneplus$. 
Moreover, we write $\cO_{1+, a}$ for the event that test $a$ in the neighbourhood of $x$ is compatible with the individual under consideration being in $\oneplus$ and $\cO_{0+, a}$ for the event that test $a$ is compatible with the considered individual being in $\zeroplus$. All that $\cO_{0+, a}$ says is that test $a$ must feature an infected individual inducing the test to be positive and thus the considered individual to be in $\zeroplus$. The same goes for $\cO_{1+, a}$.
Our test design will be denoted by $\G$. In the following, we will typically consider a reduced graph $\G'$ obtained as follows. First, remove all negative tests from $\G$ and all individuals assigned to these negative tests, i.e. the set $\zerominus$. Next, remove all individuals that now are included in at least one positive test with no other individual, i.e. the set $\oneminus$. Finally, remove any tests in the neighbourhood of $\oneminus$. Let us denote this reduced graph $\G'$ which will be the primary object of study hereafter. Note, that this graph model precisely remains after we run warning propagation on $\G$. Put differently, all removed individuals will have completely polarised marginals after running belief propagation and the remove tests are inconsequential for the posterior distribution of the remaining individuals.

\subsection{Result}

Recall from the above, the warning propagation easily identifies individuals in $\zerominus$ and $\oneminus$. Thus, those individuals with non-polarised marginals under belief propagation will be those individuals from $\zeroplus$ and $\oneplus$. Let us now state the offspring distribution for those two types of individuals.

\begin{theorem} \label{thm_dist}
Let
\begin{align*}
    p_{0,\Delta} &= \Pr \bc{x \in \zeroplus \mid x \in \zero, \abs{\partial x}=\Delta} = \bc{1+n^{-\Omega(1)}} \bc{1-e^{-d}}^{\Delta} \\
    p_{1,\Delta} &= \Pr \bc{x \in \oneplus \mid x \in \one, \abs{\partial x}=\Delta} = \bc{1+n^{-\Omega(1)}} \bc{1 - \exp\bc{-d - d \frac{1-\lambda}{\lambda} p_{0,\Delta-1}}}^{\Delta}
\end{align*} 
where $ \Erw[\Delta]=cd/\lambda$. Moreover, let 

\begin{align*}
    q_{0} &= \bc{1+n^{-\Omega(1)}} \frac{\bc{1-\exp\bc{-d p_{1,\Delta-1}}} \exp \bc{-d\bc{1-p_{1,\Delta-1}}}} {\bc{1-\exp(-d)}} \\
    q_{1} &= \bc{1+n^{-\Omega(1)}} \frac{\bc{1-\exp\bc{-d \frac{1-\lambda}{\lambda} p_{0,\Delta-1} - d p_{1,\Delta-1}}} \exp \bc{-d\bc{1-p_{1,\Delta-1}}}}{1-\exp(-d)}.
\end{align*}
For an individual $x$, the following holds for $\partial^{(2]} x$ under $\G'$.\\\\
1. Given $x\in V_0^+$ we have

\begin{align*}
    \zeroplusi{2} &\xrightarrow{d} \sum_{i=1}^{\Delta} \vX_i \Po\bc{\frac{1-\lambda}{\lambda}d p_{0,\Delta-1}} \qquad \text{as} \quad n \to \infty \\
    \oneplusi{2} &\xrightarrow{d} \sum_{i=1}^{\Delta} \vX_i \Po_{\geq 1 }\bc{d p_{1, \Delta-1}} \qquad \text{as} \quad n \to \infty.
\end{align*}
with $(\vX_i)_{i \geq 1}$ being a sequence of independent Bernoulli random variables with parameter $q_0$.\\\\
2. Given $x\in V_1^+$ we have
\begin{align*}
    \zeroplusi{2}, \oneplusi{2} \xrightarrow{d} \sum_{i=1}^\Delta \vX_i \vY_i \qquad \text{as} \quad n \to \infty.
\end{align*}
with $(\vX'_i)_{i \geq 1}$ being a sequence of independent Bernoulli random variables with parameter $q_1$ and furthermore, $(\vY_i)_{i \geq 1}$ being a sequence of independent random variables with support in $(\NN_0)^2$ and probability density function for any $i \in \NN$ and $j,k \geq 0$ given by
\begin{align*}
    \Pr \bc{\vY_i = (j,k)} = \frac{\vecone \cbc{j+k>0} \Pr \bc{\vA=j} \Pr \bc{\vB =k}}{1-\Pr \bc{\vA=0} \Pr \bc{\vB=0}}
\end{align*}
where 
\begin{align*}
    \vA &\sim \Po\bc{\frac{1-\lambda}{\lambda} d p_{0, \Delta-1}} \\
    \vB &\sim \Po\bc{d p_{1, \Delta-1}}.
\end{align*}

\end{theorem}


\section{Proof  Outline}




%

\subsection{Getting started}

It is an immediate consequence of our test design $\G$ that the number of individuals per test follows a $\Po\bc{d/\lambda}$-distribution and that the choice $d=\log(2)$ maximises the entropy of the test results. Since individuals are assigned to tests independently, we can characterise the distribution of infected and uninfected individuals as follows.
\begin{fact} \label{fact_dist}
For the number of uninfected and infected individuals in any test $a$, we have
\begin{align*}
    \abs{\partial a\cap V_0} \sim \Bin \bc{n-k, d/k} \qquad \text{and} \qquad \abs{\partial a\cap V_1} \sim \Bin \bc{k, d/k}.
\end{align*}

\end{fact}

\begin{lemma} \label{lem_prob_test}
For any individual $x$ and test $a$ we have
\begin{align*}
    \Pr \bc{x \in \partial a} = d/k
\end{align*}
\end{lemma}

\begin{proof}
A simple calculation reveals
\begin{align*}
    \Pr \bc{x \in \partial a} = 1 - \Pr \bc{x \notin \partial a} = 1 - \binom{m-1}{\Delta} \binom{m}{\Delta}^{-1} = \Delta/m = d/k.
\end{align*}
as claimed.
\end{proof}

\begin{lemma} \label{lem_zeroplus}
We have
\begin{align*}
    p_{0,\Delta} := \Pr \bc{x \in \zeroplus \mid x \in \zero,\abs{\partial x}=\Delta} = \bc{1+n^{-\Omega(1)}} \bc{1-e^{-d}}^{\Delta}
\end{align*}
\end{lemma}

\begin{proof}
Infected individuals are assigned to tests mutually independent. In combination with \Lem~\ref{lem_prob_test}, we find for a test $a$ that
\begin{align*}
    \Pr \bc{\hat \sigma_a = 1} = 1- \Pr \bc{\Bin \bc{k, d/k} = 0} = \bc{1+n^{-\Omega(1)}} \bc{1-e^{-d}}.
\end{align*}
Thus, the expected number of positive tests in the neighbourhood of any uninfected individual $x \in \zero$ is 
\begin{align*}
    \bc{1+n^{-\Omega(1)}} \bc{1-e^{-d}} \Delta.   
\end{align*}
The lemma now follows readily from the fact that $\Delta = \Theta(1)$ as $n \to \infty$. 
\end{proof}

\begin{lemma} \label{lem_oneplus}
We have
\begin{align*}
    p_{1,\Delta} := \Pr \bc{x \in \oneplus \mid x \in \one} = \bc{1+n^{-\Omega(1)}} \bc{1 - \exp\bc{-d - d \frac{1-\lambda}{\lambda} p_{0,\Delta-1}}}^\Delta
\end{align*} 
where $\Erw[\Delta]=cd/\lambda$.
\end{lemma}

\begin{proof}
By \Lem~\ref{lem_zeroplus} and the fact that the maximum number of uninfected individuals per test is $O\bc{\log n}$ with probability at least $1-o\bc{n^{-2}}$ we find for any test $a \in \partial x$ for $x \in \one$
\begin{align*}
    \Pr \bc{\partial a \cap \bc{\one \cup \zeroplus} = \emptyset} &= \Pr \bc{\Bin(k, d/k)=0} \cdot \Pr \bc{\Bin \bc{(n-k) p_{0,\Delta-1}, d/k}=0} \\
    &= \bc{1-d/k}^{k + \bc{1+n^{-\Omega(1)}} \frac{1-\lambda}{\lambda} p_{0,\Delta-1}} = \bc{1+n^{-\Omega(1)}} \exp \bc{-d - d \frac{1-\lambda}{\lambda} p_{0,\Delta-1}}
\end{align*}
The lemma now readily follows from the fact that $\Delta = \Theta(1)$ as $n \to \infty$.
\end{proof}


\subsection{Offspring for $x \in V_0^+$}

Next, we will characterize the second neighbourhood of an individual $x \in \zeroplus$ in the reduced graph $\G'$. For clarity, let us denote this individual by $x_r$.
To get started, for each test that in the neighborhood of $x_r$ we can calculate the probability that no individual from $\oneminus$ is contained in its second neighborhood. Note that we identify such tests since they are removed in $\G'$.

\begin{lemma} \label{lem_offspring_0_tests}
For any $a \in \partial x_r$, we have
\begin{align*}
    q_{0} := \Pr \bc{\oneminusi{2,a} = \emptyset \mid \cO_{0+, a}} = \bc{1+n^{-\Omega(1)}} \frac{\bc{1-\exp\bc{-d p_{1,\Delta-1}}} \exp \bc{-d\bc{1-p_{1,\Delta-1}}}} {\bc{1-\exp(-d)}}
\end{align*}
\end{lemma}

\begin{proof}
The lemma follows from a straightforward application of Bayes Theorem.
\begin{align} \label{eq_0_Bayes}
    \Pr \bc{\oneminusi{2,a} = \emptyset \mid \cO_{0+, a}} = \frac{\Pr \bc{\cO_{0+, a} \mid \oneminusi{2,a} = \emptyset} \Pr \bc{\oneminusi{2,a} = \emptyset}}{\Pr \bc{\cO_{0+, a}}}.
\end{align}
By Fact~\ref{fact_dist} and \Lem s~\ref{lem_zeroplus} and \ref{lem_oneplus} we have
\begin{align}
    \Pr \bc{\cO_{0+, a} \mid \oneminusi{2,a} = \emptyset} &= \bc{1-\Pr \bc{\oneplusi{2,a}=\emptyset}} = \bc{1+n^{-\Omega(1)}} \bc{1-\exp\bc{-d p_{1,\Delta-1}}} \label{eq_0_1} \\
    \Pr \bc{\oneminusi{2,a} = \emptyset} &= \bc{1+n^{-\Omega(1)}} \exp \bc{-d\bc{1-p_{1,\Delta-1}}} \label{eq_0_2} \\
    \Pr \bc{\cO_{0+, a}} &= 1-\Pr \bc{\onei{2, a}=\emptyset} = \bc{1+n^{-\Omega(1)}} \bc{1-\exp(-d)} \label{eq_0_3}
\end{align}
Plugging \eqref{eq_0_1}--\eqref{eq_0_3} into \eqref{eq_0_Bayes} yields the desired expression.
\end{proof}

Thus, a test in the neighborhood of $x_r$ remains unexplained with $q_0$. Using this insight, we obtain the following two results for the offspring distribution in the second neighbourhood of $x_r$.

\begin{lemma} \label{lem_offspring_0_0}
Let $(\vX_i)_{i \geq 1}$ be a sequence of independent Bernoulli random variables with parameter $q_0$. Given $\cO_{0+}$ we have
\begin{align*}
    \zeroplusi{2} \xrightarrow{d} \sum_{i=1}^\Delta \vX_i \Po\bc{\frac{1-\lambda}{\lambda}d p_{0,\Delta-1}} \qquad \text{as} \quad n \to \infty.
\end{align*}
\end{lemma}

\begin{proof}
Given $\cO_{0+}$, we know that any test $a \in \partial x_r$ contains an infected individual and is thus positive. Starting from the distribution of uninfected individuals in test $a$ from Fact~\ref{fact_dist} and using \Lem~\ref{lem_zeroplus} and the fact that the maximum test degree is $O(\log n)$ with probability at least $1-o\bc{n^{-2}}$, we find
\begin{align*}
   \Bigg\{  \Big\{\zeroplusi{2,a} \mid \cO_{0+}, \oneminusi{2,a}\Big\} = \emptyset\Bigg\} \xrightarrow{d} \Po \bc{\frac{1-\lambda}{\lambda} d p_{0,\Delta-1}}.  
\end{align*}
as $n\to\infty$. The lemma now follows readily from the fact that warning propagation removes tests featuring individuals from $\oneminus$ and that we have $\Delta=\Theta(1)$.
\end{proof}

\begin{lemma} \label{lem_offspring_0_1}
Let $(\vX_i)_{i \geq 1}$ be a sequence of independent Bernoulli random variables with parameter $q_0$. Given $x \in V_0^+$ we have
\begin{align*}
    \oneplusi{(2)} \xrightarrow{d} \sum_{i=1}^\Delta \vX_i \Po_{\geq 1 }\bc{d p_{1, \Delta-1}} \qquad \text{as} \quad n \to \infty.
\end{align*}
\end{lemma}

\begin{proof}
Another application of Bayes Theorem reveals for any test $a \in \partial x_r$ and any $i \geq 0$
\begin{align}
    \Pr \bc{\abs{\oneplusi{2,a}}=i \Big|\cO_{0+, a}, \Big(\oneminusi{2,a}=\emptyset\Big)} &= \frac{\Pr \bc{\cO_{0+,a} \mid \abs{\oneplusi{2,a}}=i, \oneminusi{2,a}=\emptyset} \Pr \bc{\abs{\oneplusi{2,a}}=i \mid \oneminusi{2,a}=\emptyset}}{\Pr \bc{\cO_{0+,a} \mid \oneminusi{2,a}=\emptyset}} \notag \\
    &= \frac{\vecone \cbc{i > 0} \Pr \bc{\abs{\oneplusi{2,a}}=i}}{1-\Pr\bc{\abs{\oneplusi{2,a}}=0}} \label{eq_0_4}.
\end{align}
Indeed, \eqref{eq_0_4} is the probability density function of the distribution $\Po_{\geq 1}\bc{d q_{1,\Delta-1}}$. Thus, the lemma follows from the fact that warning propagation removes tests featuring individuals from $\oneminus$ and $\Delta=\Theta(1)$.
\end{proof}

\subsection{Offspring for $x \in V_1^+$}

We adopt the same notation as above. To be precise, we will write $\cO_{1+}$ to denote that the considered individual is from $\oneplus$. Moreover, we write $\cO_{1+, a}$ for the event that test $a$ is compatible with the considered individual being in $\oneplus$, i.e. the test featuring at least one other infected individual or individual from $\zeroplus$.

\begin{lemma} \label{lem_offspring_1_tests}
For any $a \in \partial x_r$, we have
\begin{align*}
    q_{1} := \Pr \bc{\oneminusi{2,a} = \emptyset \mid \cO_{1+, a}} = \bc{1+n^{-\Omega(1)}} \frac{\bc{1-\exp\bc{-d \frac{1-\lambda}{\lambda} p_{0,\Delta-1} - d p_{1,\Delta-1}}} \exp \bc{-d\bc{1-p_{1,\Delta-1}}}}{1-\exp(-d)}
\end{align*}
\end{lemma}

\begin{proof}
By Bayes Theorem we have
\begin{align} \label{eq_1_Bayes} 
    \Pr \bc{\oneminusi{2,a} = \emptyset \mid \cO_{1+, a}} &= \frac{\Pr \bc{\cO_{1+, a} \mid \oneminusi{2,a} = \emptyset} \Pr \bc{\oneminusi{2,a} = \emptyset}}{\Pr \bc{\cO_{1+, a}}}.
\end{align}
For each of these terms, we find
\begin{align}
    \Pr \bc{\cO_{1+, a} \mid \oneminusi{2,a} = \emptyset} &= 1-\Pr\bc{\zeroplusi{2,a}=\emptyset}\Pr\bc{\zeroplusi{2,a}=\emptyset} \notag\\
    &= \bc{1+n^{-\Omega(1)}} \bc{1-\exp\bc{-d \frac{1-\lambda}{\lambda} p_{0,\Delta-1} - d p_{1,\Delta-1}}} \label{eq_1_1} \\
    \Pr \bc{\oneminusi{2,a} = \emptyset} &= \bc{1+n^{-\Omega(1)}} \exp \bc{-d\bc{1-p_{1,\Delta-1}}} \label{eq_1_2} \\
    \Pr \bc{\cO_{1+, a}} &= 1-\Pr \bc{\onei{2, a}=\emptyset} = \bc{1+n^{-\Omega(1)}} \bc{1-\exp(-d)} \label{eq_1_3}
\end{align}
Plugging \eqref{eq_1_1}--\eqref{eq_1_3} into \eqref{eq_1_Bayes} yields the desired expression.
\end{proof}

Next, let us specify the offspring distribution for $\zeroplusi{2}$ and $\oneplusi{2}$. In contrast to having $x \in \zeroplus$ as the starting individual, the distribution of both types is not independent and we have to specify a joint distribution.

\begin{lemma} \label{lem_offspring_1_01}
Given $\cO_{1+}$ we have
\begin{align*}
    \zeroplusi{2}, \oneplusi{2} \xrightarrow{d} \sum_{i=1}^\Delta \vX_i \vY_i \qquad \text{as} \quad n \to \infty 
\end{align*}
with $(\vX_i)_{i \geq 1}$ be a sequence of independent Bernoulli random variables with parameter $q_1$ and $(\vY_i)_{i \geq 1}$ being a sequence of independent random variables with support in $(\NN_0)^2$ and probability density function for any $i \in \NN$ and $j,k \geq 0$ given by
\begin{align*}
    \Pr \bc{\vY_i = (j,k)} = \frac{\vecone \cbc{j+k>0} \Pr \bc{\vA=j} \Pr \bc{\vB =k}}{1-\Pr \bc{\vA=0} \Pr \bc{\vB=0}}
\end{align*}
where 
\begin{align*}
    \vA &\sim \Po\bc{\frac{1-\lambda}{\lambda} d p_{0, \Delta-1}} \\
    \vB &\sim \Po\bc{d p_{1, \Delta-1}}.
\end{align*}

\end{lemma}

\begin{proof}
For any test $a \in \partial x_r$ and $i,j \geq 0$, we find by Bayes Theorem
\begin{align*}
    &\Pr \bc{\abs{\zeroplusi{2,a}}=i, \abs{\oneplusi{2,a}}=j \mid \cO_{1+, a}, \oneminusi{2,a}=\emptyset}\\
    &\qquad = \frac{\Pr \bc{\cO_{1+, a} \mid \abs{\zeroplusi{2,a}}=i, \abs{\oneplusi{2}}=j, \oneminusi{2,a}=\emptyset} \Pr \bc{\abs{\zeroplusi{2,a}}=i} \Pr \bc{\Pr \bc{\abs{\oneplusi{2}}=j}}}{1-\Pr \bc{\zeroplusi{2,a}=\emptyset} \Pr \bc{\oneplusi{2,a}=\emptyset}} \\
    &= \frac{\vecone \cbc{i+j>0} \Pr \bc{\abs{\zeroplusi{2,a}}=i} \Pr \bc{\abs{\oneplusi{2}}=j}}{1-\Pr \bc{\zeroplusi{2,a}=\emptyset} \Pr \bc{\oneplusi{2,a}=\emptyset}}.
\end{align*}
For the remaining two terms, we readily find
\begin{align*}
    \Pr \bc{\abs{\zeroplusi{2,a}}=i} &= \bc{1+n^{-\Omega(1)}} \Pr \bc{\Po\bc{\frac{1-\lambda}{\lambda} d p_{0, \Delta-1}} = i} \\
    \Pr \bc{\abs{\oneplusi{2,a}}=j} &= \bc{1+n^{-\Omega(1)}} \Pr \bc{\Po\bc{d p_{1, \Delta-1}}}
\end{align*}
\end{proof}

\begin{proof}[Proof of \Thm~\ref{thm_dist}]
The theorem follows from \Lem s~\ref{lem_offspring_0_0}, \ref{lem_offspring_0_1} and \ref{lem_offspring_1_01}.
\end{proof}

\bibliographystyle{plain}
\bibliography{bibliography}

\end{document}